\documentclass{tlp}
\usepackage{aopmath}

\newtheorem{example}{Example}

\newcommand{\systemname}[1]{\emph{#1}}

\newcommand{\alphabet}{\mathcal{A}}
\newcommand{\calphabet}{\mathcal{C}}
\newcommand{\assignment}{\mathbf{A}}
\newcommand{\cassignment}{A}

\newcommand{\tass}[1]{\mathbf{T}#1}
\newcommand{\fass}[1]{\mathbf{F}#1}
\newcommand{\Tass}{\assignment^\mathbf{T}}
\newcommand{\Fass}{\assignment^\mathbf{F}}
\newcommand{\atom}[1]{atom(#1)}
\newcommand{\head}[1]{head(#1)}
\newcommand{\body}[1]{body(#1)}
\newcommand{\dneg}{not\ }
\newcommand{\domain}[1]{dom(#1)}
\newcommand{\range}[1]{range(#1)}
\newcommand{\scope}[1]{scope(#1)}

\newcommand{\encsup}{$S$}
\newcommand{\encbou}{$B$}
\newcommand{\encran}{$R$}
\newcommand{\encbouh}[1]{\encbou$_{#1}$}
\newcommand{\encranh}[1]{\encran$_{#1}$}

\begin{document}

\title[A Translational Approach to Constraint Answer Set Solving]
{A Translational Approach to\\ Constraint Answer Set Solving}

\author[C. Drescher and T. Walsh]
{CHRISTIAN DRESCHER\thanks{Part of this work was performed when Christian Drescher was studying at the New University of Lisbon, Portugal.} \\
Vienna University of Technology, Austria
\and TOBY WALSH \\
NICTA and University of New South Wales, Australia
}

\submitted{26 January 2010}
\revised{12 May 2010}
\accepted{21 March 2010}
\pubyear{2010}

\maketitle

\begin{abstract}
We present a new approach to enhancing Answer Set Programming (ASP) with Constraint Processing techniques which allows for solving interesting Constraint Satisfaction Problems in ASP. We show how constraints on finite domains can be decomposed into logic programs such that unit-propagation achieves arc, bound or range consistency.
Experiments with our encodings demonstrate their computational impact.
\end{abstract}

\begin{keywords}
 answer set programming, constraint processing, decomposition
\end{keywords}

\section{Introduction}

Answer Set Programming (ASP; \citeNP{baral03}) has been put forward as a powerful paradigm to solve Constraint Satisfaction Problems (CSP) in \cite{niemela99a}. Indeed, ASP has been shown to be a useful in various applications, among them planning \cite{lifschitz99b}, model checking \cite{helnie03a}, and bio-informatics \cite{bachtrtrjobe04a}, and decision support for NASA shuttle controllers \cite{nobagewaba01a}. It combines an expressive but simple modelling language with high-performance solving capacities. In fact, modern ASP solvers, such as \systemname{clasp} \cite{gekanesc07b}, compete with the best Boolean Satisfiability (SAT; \citeNP{bihemawa09a}) solvers.
An empirical comparison of the performance of ASP and traditional Constraint Logic Programming (CLP; \citeNP{jama94a}) on solving CSP conducted by \citeN{dofopo05a} shows ASP encodings to be more compact, more declarative, and highly competitive, but also revealed shortcomings: non-Boolean constructs, like resources or functions over finite domains, in particular global constraints, are more naturally modelled and efficiently handled by Constraint Processing (CP; \citeNP{dechter03,robewa06a}) systems.

This led to the integration of CP techniques into ASP.
Similar to Satisfiability Modulo Theories (SMT; \citeNP{niolti06a}), the key idea of an \emph{integrative} approach is to incorporate theory-specific predicates into propositional formulas, and extending an ASP solver's decision engine for a more high-level proof procedure.
Recent work on combining ASP with CP was conducted in \cite{baboge05a,megezh08a,melgel08a} and \cite{geossc09a}. While \citeANP{melgel08a} both view ASP and CP solvers as blackboxes, \citeANP{geossc09a} embed a CP solver into an ASP solver adding support for advanced backjumping and conflict-driven learning techniques.
\citeN{ba09a} and \citeN{jaoijani09a} cut ties to ad-hoc ASP and CP solvers, and principally support global constraints.
\citeN{padoporo09a} put further emphasis on handling constraint variables with large domains, and presented a strategy which only consider parts of the model that actively contribute in supporting constraint answer sets.
However, each system has a subset of the following limitations: either they are tied to particular ASP and CP solvers, or the support for global constraints is limited, or communication between the ASP and CP solver is restricted.

This paper introduces a \emph{translational} approach to Constraint Answer Set Solving rather than an integrative one. Motivated by the success of SAT-based constraint solvers, such as the award-winning system \systemname{Sugar} \cite{tatakiba06a}, we show how to enhance ASP with Constraint Processing techniques through translation to ASP.
A first study was conducted in \citeN{gehiscth09a} with the system \systemname{xpanda} for representing multi-valued propositions in ASP.
One of the key contributions of our work is an investigation of constraint decomposition techniques in the new field of Constraint Answer Set Programming, illustrated on the popular \emph{all-different} constraint. The resulting approach has been implemented in the new preprocessor \systemname{inca}. Empirical evaluation demonstrates its computational potential.

The remainder of this paper is organized as follows. We start by giving the background notions of ASP and Constraint Satisfaction. Various generic ASP encodings of constraints on finite domains and proofs of their properties are given in Section \ref{sec:enc}. In Section \ref{sec:exp}, we empirically evaluate our approach and compare to existing research. Section \ref{sec:con} draws conclusions.

\section{Background}
\subsection{Answer Set Programming}

A \emph{(normal) logic program} over a set of primitive propositions $\alphabet$ is a finite set of \emph{rules} of the form
\[
h \leftarrow a_1 , \dots , a_m, \dneg a_{m+1} , \dots , \dneg a_n
\]
where $0 \leq m \leq n$ and $h, a_i \in \alphabet$ is an \emph{atom} for $1 \leq i \leq n$.
A \emph{literal} $\hat{a}$ is an atom $a$ or its default negation $\dneg a$.
For a rule~$r$, let $\head{r} = h$ be the \emph{head} of $r$ and $\body{r} = \{a_1 , \dots , a_m, \dneg a_{m+1} , \dots , \dneg a_n\}$ the \emph{body} of $r$. Furthermore, define $\body{r}^{+} = \{a_1 , \dots , a_m\}$ and $\body{r}^{-} = \{a_{m+1} , \dots , a_n\}$. The set of atoms occurring in a logic program $\Pi$ is denoted by $\atom{\Pi}$, and the set of bodies in $\Pi$ is $\body{\Pi} = \{ \body{r} \mid r \in \Pi \}$. For regrouping bodies sharing the same head $a$, define $\body{a} = \{ \body{r} \mid r \in \Pi,\ \head{r} = a \}$.

The semantics of a program is given by its answer sets. A set $X \subseteq \alphabet$ is an \emph{answer set} of a logic program $\Pi$ over $\alphabet$, if $X$ is the $\subseteq$-minimal model of the \emph{reduct} \cite{gellif88b}
\[
\Pi^{X} = \{ \head{r} \leftarrow \body{r}^{+} \mid r \in \Pi,\ \body{r}^{-} \cap X = \emptyset\}.
\]
The semantics of important extensions to normal logic programs, such as choice rules, integrity and cardinality constraints, is given through program transformations that introduce additional propositions (cf. \citeNP{siniso02a}).
A \emph{choice rule} allows for the non-deterministic choice over atoms in $\{h_1, \dots, h_k\}$ and has the following form:
\[
\{h_1, \dots, h_k\} \leftarrow a_1 , \dots , a_m, \dneg a_{m+1} , \dots , \dneg a_n
\]
An \emph{integrity constraint}
\[
\leftarrow a_1 , \dots , a_m, \dneg a_{m+1} , \dots , \dneg a_n
\]
is an abbreviation for a rule with an unsatisfiable head, and thus forbids its body to be satisfied in any answer set.
A~\emph{cardinality constraint}
\[
\leftarrow k \{\hat{a}_1 , \dots, \hat{a}_n\}
\]
is interpreted as no answer set satisfies $k$ literals of the set $\{\hat{a}_1 , \dots, \hat{a}_n\}$. It can be transformed into ${n \choose k}$ integrity constraints $r$ such that $\body{r} \subseteq \{\hat{a}_1 , \dots, \hat{a}_n\}$ and $|\body{r}| = k$.
\citeANP{siniso02a} provide a transformation that needs just $\mathcal{O}(nk)$ rules, introducing atoms $l(\hat{a}_i,j)$ to represent the fact that at least $j$ of the literals with index $\geq i$, i.e. the literals in $\{ \hat{a}_i, \dots, \hat{a}_n \}$,
 are in a particular answer set candidate. Then, the cardinality constraint can be encoded by an integrity constraint $\leftarrow l(\hat{a}_1,k)$ and the three following rules, where $1 \leq i \leq n$ and $1 \leq j \leq k$:
\[
l(\hat{a}_i,j) \leftarrow l(\hat{a}_{i+1},j) \qquad\quad
l(\hat{a}_i,j+1) \leftarrow \hat{a}_i, l(\hat{a}_{i+1},j)\qquad\quad
l(\hat{a}_i,1) \leftarrow \hat{a}_i
\]
Notice that both transformations are modular. Alternatively, modern ASP solvers also incorporate propagators for cardinality constraints that run in $\mathcal{O}(n)$.

\subsection{Nogoods}

We want to view inferences in ASP as unit-propagation on nogoods. 
Following \citeN{gekanesc07a}, inferences in ASP rely on atoms and program rules, which can be expressed by using atoms and bodies. Thus, for a program~$\Pi$, the \emph{domain} of Boolean assignments~$\assignment$ is fixed to $\domain{\assignment} = \atom{\Pi} \cup \body{\Pi}$.

Formally, a Boolean \emph{assignment} $\assignment$ is a set $\{ \sigma_1, \dots, \sigma_n \}$ of \emph{signed literals}~$\sigma_i$ for $1 \leq i \leq n$ of the form $\tass{a}$ or $\fass{a}$ where $a \in \domain{\assignment}$. $\tass{a}$ expresses that $a$ is assigned \emph{true} and $\fass{a}$ that it is \emph{false} in $\assignment$. (We omit the attribute \emph{Boolean} for assignments whenever clear from the context.) The complement of a signed literal~$\sigma$ is denoted by $\overline{\sigma}$, that is $\overline{\tass{a}} = \fass{a}$ and $\overline{\fass{a}} = \tass{a}$.
In the context of ASP, a \emph{nogood} \cite{dechter03} is a set $\delta = \{ \sigma_1, \dots, \sigma_n \}$ of signed literals, expressing a constraint violated by any assignment~$\assignment$ such that $\delta \subseteq \assignment$.
For a nogood $\delta$, a signed literal $\sigma \in \delta$, and an assignment $\assignment$, we say that $\delta$ is \emph{unit} and $\overline{\sigma}$ is \emph{unit-resulting} if $\delta \setminus \assignment = \{\sigma\}$.
Let $\Tass = \{ a \in \domain{\assignment} \mid \tass{a} \in A \}$ the set of true propositions and $\Fass = \{ a \in \domain{\assignment} \mid \fass{a} \in A \}$ the set of false propositions. A \emph{total} assignment, that is $\Tass \cup \Fass = \domain{\assignment}$ and $\Fass \cup \Tass = \emptyset$, is a \emph{solution} for a set $\Delta$ of nogoods if $\delta \not\subseteq \assignment$ for all $\delta \in \Delta$.

As shown in \citeN{lee05a}, the answer sets of a logic program $\Pi$ correspond to the models of the completion of $\Pi$ that satisfy the loop formulas of all non-empty subsets of $\atom{\Pi}$. For $\beta = \{ a_1 , \dots , a_m, \dneg a_{m+1} , \dots , \dneg a_n \} \in \body{\Pi}$, define
\[
\Delta_\beta = \left\{ \begin{array}{l}
\{\tass{a_1}, \dots, \tass{a_m}, \fass{a_{m+1}}, \dots \fass{a_n}, \fass{\beta} \}, \\
\{\fass{a_1}, \tass{\beta}\}, \dots, \{\fass{a_m}, \tass{\beta}\}, \{\tass{a_{m+1}}, \tass{\beta}\}, \dots, \{\tass{a_n}, \tass{\beta}\}
\end{array} \right\}.
\]
Intuitively, the nogoods in $\Delta_\beta$ enforce the truth of body~$\beta$ iff all its literals are satisfied.
For an atom $a \in \atom{\Pi}$ with $\body{a} = \{\beta_1, \dots, \beta_k\}$, let
\[
\Delta_a = \left\{ \begin{array}{l}
\{\fass{\beta_1}, \dots, \fass{\beta_k}, \tass{a} \}, \\
\{\tass{\beta_1}, \fass{a}\}, \dots, \{\tass{\beta_k}, \fass{a}\}
\end{array} \right\}.
\]
Then, the solutions for $\Delta_\Pi = \bigcup_{\beta \in \body{\Pi}} \Delta_\beta \cup \bigcup_{a \in \atom{\Pi}} \Delta_a$ correspond to the models of the completion of $\Pi$. Loop formulas, expressed in the set of nogoods~$\Lambda_\Pi$, have to be added to establish full correspondence to the answer sets of $\Pi$.
Typically, solutions for $\Delta_\Pi \cup \Lambda_\Pi$ are computed by applying \emph{Conflict-Driven Nogood Learning} (CDNL; \citeNP{gekanesc07a}). This combines search and propagation by recursively assigning the value of a proposition and using unit-propagation (Fig. \ref{alg:up}) to determine logical consequences of an assignment \cite{mitchell05a}.
\begin{figure}
\figrule
\[
\begin{array}{ll}
$\textbf{Input}:$ & $A set $\nabla$ of nogoods, and an assignment $\assignment$.$\\
$\textbf{Output}:$ & $An extended assignment, and a status (either $conflict$ or $success$).$\\
\end{array}
\]
\[
\begin{array}{l}
$\textbf{repeat}$\\
\qquad$\textbf{if} $\delta \subseteq \assignment $ for some $ \delta \in \nabla $ \textbf{then}$\\
\qquad\qquad$\textbf{return} $(\assignment,conflict)$;$\\
\qquad\Sigma\leftarrow\{\delta\in\nabla\mid\delta\setminus\assignment=\{\sigma\},\overline{\sigma}\not\in\assignment\}$;$\\
\qquad$\textbf{if} $\Sigma\neq\emptyset$ \textbf{then let} $\sigma\in\delta\setminus\assignment$ for some $\delta\in\Sigma$ \textbf{in}$\\
\qquad\qquad\assignment\leftarrow\assignment\cup(\overline{\sigma})$;$\\
$\textbf{until} $\Sigma=\emptyset $;$\\
$\textbf{return} $(\assignment,success)$;$
\end{array}
\]
\figrule
\vspace{-1\baselineskip}
\caption{The unit-propagation algorithm.}\label{alg:up}
\end{figure}
\begin{example}
Consider the set of nogoods $\nabla = \{ \{\tass{a_1}, \fass{a_2}, \tass{a_3}, \tass{a_4}\}, \{\fass{a_1}, \tass{a_4}\}$, $\{\fass{a_3}, \tass{a_4}\} \}$ and the assignment $\assignment = \{\tass{a_4}\}$. Unit-propagation extends $\assignment$ by $\{\tass{a_1}$, $\tass{a_2}$, $\tass{a_3}\}$.
\end{example}

\subsection{Constraint Satisfaction and Consistency}

A \emph{Constraint Satisfaction Problem} is a triple $(V,D,C)$ where $V$ is a set of \emph{variables} $V = \{v_1, \dots , v_n\}$, $D$ is a set of finite \emph{domains} $D=\{D_1, \dots , D_n\}$ such that each variable~$v_i$ has an associated domain $\domain{v_i} = D_i$, and $C$ is a set of \emph{constraints}. Following \citeN{robewa06a},
a constraint~$c$ is a pair~$(R_S,S)$ where $R_S$ is a $k$-ary \emph{relation} on the variables in $S \subseteq V^k$, called the \emph{scope} of $c$. In other words, $R_S$ is a subset of the Cartesian product of the domains of the variables in $S$. To access the relation and the scope of $c$ define $\range{c} = R_S$ and $\scope{c} = S$. For a \emph{(constraint variable) assignment} $\cassignment : V \to \bigcup_{v \in V} dom(v)$ and a constraint $c = (R_S, S)$ with $S = (v_1, \dots, v_k)$, define $\cassignment(S) = (\cassignment(v_1), \dots, \cassignment(v_k))$, and call $c$ \emph{satisfied} if $\cassignment(S) \in \range{c}$. A binary constraint~$c$ has $|scope(c)|=2$. For example,
$v_1 \neq v_2$ ensures that $v_1$ and $v_2$ take different
values. A global (or $n$-ary) constraint~$c$ has parametrized
scope. For example, 
the \emph{all-different} constraint ensures that
a set of variables, $\{v_1,\ldots,v_n\}$ take all different values. 
This can be decomposed into $O(n^2)$ binary
constraints, $v_i \neq v_j$ for $i<j$. However, as 
we shall see, such decomposition can hinder inference. 
An assignment~$\cassignment$ is a \emph{solution} for a CSP iff it satisfies all constraints in $C$.

Constraint solvers typically use \emph{backtracking search} to explore the space of partial assignments, and prune it by applying propagation algorithms that enforce a \emph{local consistency}
property on the constraints after each assignment.
A binary constraint~$c$ is called \emph{arc consistent} iff when a variable~$v_1 \in \scope{c}$ is assigned any value~$d_1 \in \domain{v_1}$, there exists a consistent value~$d_2 \in \domain{v_2}$ for the other variable~$v_2$.
An $n$-ary constraint~$c$ is \emph{hyper-arc consistent} or \emph{domain consistent} iff when a variable~$v_i \in \scope{c}$ is assigned any value~$d_i \in \domain{v_i}$, there exist compatible values in the domains of all the other variables~$d_j \in \domain{v_j}$ for all $1 \leq j \leq n,\ j \neq i$ such that $(d_1, \dots, d_n) \in \range{c}$.

The concepts of bound and range consistency are defined for constraints on ordered intervals.
Let $min(D_i)$ and $max(D_i)$ be the minimum value and maximum value of the domain~$D_i$. A constraint~$c$ is \emph{bound consistent} iff when a variable~$v_i$ is assigned $d_i \in \{min(\domain{v_i}), max(\domain{v_i})\}$ (i.e. the minimum or maximum value in its domain), there exist compatible values between the minimum and maximum domain value for all the other variables in the scope of the constraint. Such an assignment is called a \emph{bound support}. A constraint is \emph{range consistent} iff when a variable is assigned any value in its domain, there exists a bound support. Notice that range consistency is in between domain and bound consistency, where domain consistency is the strongest of the four formalisms.

\subsection{Constraint Answer Set Programming}

Following \citeN{geossc09a}, a \emph{constraint logic program}~$\Pi$ is defined as logic programs over an extended alphabet distinguishing regular and constraint atoms, denoted by $\alphabet$ and $\calphabet$, respectively, such that $\head{r} \in \alphabet$ for each $r \in \Pi$.
Constraint atoms are identified with constraints via a function $\gamma : \calphabet \to C$, and furthermore, define $\gamma(C') = \{ \gamma(c) \mid c \in C' \}$ for $C' \subseteq C$. For a (constraint variable) assignment~$\cassignment$ define the set of constraints satisfied by $\cassignment$ as
$
sat_C(\cassignment) = \{ c \mid \cassignment(\scope{c}) \in \range{c},\ c \in C\},
$
and the \emph{constraint reduct} as 
\[
\begin{array}{l}
\Pi^{\cassignment} = \{ \head{r} \leftarrow \body{r}|_\alphabet \mid r \in \Pi,\hfill\\ \qquad\qquad\qquad\gamma(\body{r}^{+}|_\calphabet) \subseteq sat_C(\cassignment),\ \gamma(\body{r}^{-}|_\calphabet) \cap sat_C(\cassignment) = \emptyset\}.
\end{array}
\]
Then, a set $X \subseteq \alphabet$ is a \emph{constraint answer set} of $\Pi$ with respect to $\cassignment$, if $X$ is an answer set of $\Pi^\cassignment$.

In the translational approach to Constraint Answer Set Solving, a constraint logic program is compiled into a (normal) logic program by adding an ASP decomposition of all constraints comprised in the constraint logic program. The constraint answer sets can then be obtained by applying the same algorithms as for calculating answer sets, e.g. CDNL. Since all variables will be shared between constraints, nogood learning techniques as in CDNL exploit constraint interdependencies. This can improve propagation between constraints.

\section{Encoding Constraint Answer Set Programs} \label{sec:enc}

In this section we explain how to translate constraint logic programs with multi-valued propositions into a (normal) logic program. There are a number of choices of how to encode constraints on multi-valued propositions, e.g. a constraint variable~$v$, taking values out of a pre-defined finite domain, $\domain{v}$. In what follows, we assume $\domain{v} = \lbrack 1, d\rbrack$ for all $v \in V$ to save the reader from multiple superscripts.

\subsection{Direct Encoding}

A popular choice is called the \emph{direct encoding} \cite{wa00}. In the direct encoding, a propositional variable $e(v, i)$, representing $v = i$, is introduced for each value~$i$ that can be assigned to the constraint variable~$v$. Intuitively, the proposition $e(v, i)$ is true if $v$ takes the value $i$, and false if $v$ takes a value different from $i$. For each $v$, the truth-assignments of atoms $e(v, i)$ are encoded by a choice rule (1). Furthermore, there is an integrity constraint (2) to ensure that $v$ takes at least one value, and a cardinality constraint (3) that ensures that $v$ takes at most one value.
\[
\begin{array}{lr@{\ \leftarrow\ }l}
(1) & \{ e(v, 1), \dots, e(v, d) \} & \\
(2) & & \dneg e(v, 1), \dots, \dneg e(v, d) \\
(3) & & 2\ \{ e(v, 1), \dots, e(v, d) \}
\end{array}
\]
In the direct encoding,
each forbidden combination of values in a constraint is expressed by an integrity constraint. On the other
hand, when a relation is represented by allowed combinations of values, all forbidden combinations have to be deduced and translated to integrity constraints. Unfortunately, the direct encoding of constraints hinders propagation:
\begin{theorem}[\citeNP{wa00}]
Enforcing arc consistency on the binary decomposition of the original constraint prunes more values from the variables domain than unit-propagation on its direct encoding.
\end{theorem}

\subsection{Support Encoding}

The \emph{support encoding} has been proposed to
tackle this weakness \cite{gent02}. A \emph{support} for a constraint variable~$v$ to take the value~$i$ across a constraint~$c$ is the set of values $\{i_1, \dots, i_m\} \subseteq \domain{v'}$ of another variable in~$v' \in \scope{c}\setminus \{v\}$ which allow $v = i$, and can be encoded as follows, extending (1--3):
\[
\leftarrow e(v, i), \dneg e(v', i_1), \dots, \dneg e(v', i_m)
\]
This integrity constraint can be read as whenever $v = i$, then at least one of its supports must hold.
In the support encoding, for each constraint~$c$ there is one support for each pair of distinct variables $v, v' \in \scope{c}$, and for each value $i$.
\begin{theorem}[\citeNP{gent02}] \label{thm:arc}
Unit-propagation on the support encoding enforces arc consistency on the binary decomposition of the original constraint.
\end{theorem}
We illustrate this approach on an encoding of the global \emph{all-different} constraint. For variables $v, v'$ and value $i$ it is defined by the following $\mathcal{O}(n^2d)$ integrity constraints:
\[
\leftarrow e(v, i), \dneg e(v', 1), \dots, \dneg e(v', {i-1}), \dneg e(v', {i+1}), \dots, \dneg e(v', d)
\]
To keep the encoding small, we make use of the following equivalence (e)
\[
e(v', i) \equiv \dneg e(v', 1), \dots, \dneg e(v', {i-1}), \dneg e(v', {i+1}), \dots, \dneg e(v', d)
\]
covered by (2--3) and get
\[
\leftarrow e(v, i), e(v', i).
\]
Observe, that this is also
the direct encoding of the binary decomposition of the global \emph{all-different} constraint. However, this observation does not hold in general for all constraints.
As discussed in the Background section of this paper, we can express above condition as $\mathcal{O}(d)$ cardinality constraints:
\[
\begin{array}{lr@{\ \leftarrow\ }l}
(4) & & 2\ \{ e(v_1, i), \dots, e(v_n, i) \}
\end{array}
\]

\begin{corollary}
Unit-propagation on (1--4) enforces arc consistency on the binary decomposition of the global \emph{all-different} constraint in $\mathcal{O}(nd^2)$ down any branch of the search tree.
\end{corollary}
\begin{proof}
From the definition of cardinality constraints, (4) ensure that for all distinct $v,v' \in \scope{c}$ any value $i$ is not taken by both $v$ and $v'$. These integrity constraints correspond to the support encoding of the global \emph{all-different} constraint since (2--3) cover the equivalence ($e$).
By Theorem~\ref{thm:arc}, unit-propagation on this support encoding enforces arc consistency on the binary decomposition of the \emph{all-different} constraint.

For each of the $n$ variables, there are $\mathcal{O}(d)$ nogoods resulting from (1--3) that can be woken $\mathcal{O}(d)$ times down any branch of the search tree. Each propagation requires $\mathcal{O}(1)$ time. Rules (1--3) therefore take $\mathcal{O}(nd^2)$ down any branch of the search to propagate. There are $\mathcal{O}(nd)$ nogoods resulting from (4) that each take $\mathcal{O}(1)$ time to propagate down any branch of the search tree. The total running time is given by $\mathcal{O}(nd^2) + \mathcal{O}(nd) = \mathcal{O}(nd^2)$.
\end{proof}

\subsection{Range Encoding}

In the \emph{range encoding}, a propositional variable~$r(v, l, u)$ is introduced for all $\lbrack l, u \rbrack \subseteq \lbrack 1, d \rbrack$ to represent whether the value of~$v$ is between $l$ and $u$. For each range~$\lbrack l, u \rbrack$, the following $\mathcal{O}(nd^2)$ rules encode $v \in \lbrack l , u \rbrack$ whenever it is safe to assume that $v \not\in \lbrack 1, l-1 \rbrack$ and $v \not\in \lbrack u+1, d\rbrack$, and enforce a consistent set of ranges such that $v \in \lbrack l, u\rbrack \Rightarrow v \in \lbrack l-1, u\rbrack \land v \in \lbrack l, u+1\rbrack$:
\[
\begin{array}{lr@{\ \leftarrow\ }l}
(5) & r(v, l, u) & \dneg r(v, 1, l-1), \dneg r(v, u+1, d) \\
(6) & & r(v, l-1, u), \dneg r(v, l, u) \\
(7) & & r(v, l, u+1), \dneg r(v, l, u)
\end{array}
\]
Constraints are encoded into integrity constraints representing conflict regions. When the combination $v_1 \in \lbrack l_1, u_1\rbrack, \dots, v_n \in \lbrack l_n, u_n\rbrack$ violates the constraint, the following rule is added:
\[
\leftarrow r(v_1, l_1, u_1), \dots, r(v_n, l_n, u_n)
\]
\begin{theorem} \label{thm:rng}
Unit-propagation on the range encoding enforces range consistency on the original constraint.
\end{theorem}
\begin{proof}
Suppose we have a set of ranges on the domains of the constraint variables in which no unit-propagation is possible and no domain is empty.
Consider any constraint variables~$v_1, \dots, v_i, \dots, v_n$ and value~$d_i$ such that there is no bound support in $v_1, \dots, v_{i-1}, v_{i+1}, \dots v_n$ for $v_i = d_i$, i.e. there are no compatible values for the other variables~$v_j$ distinct from $v_i$ where $v_j \in \lbrack l_j, u_j \rbrack$. Hence, all instantiations such that $v_i = d_i$ are in a conflict region $v_1 \in \lbrack l_1, u_1 \rbrack \subseteq \lbrack l'_1, u'_1 \rbrack, \dots, v_n \in \lbrack l_n, u_n \rbrack \subseteq \lbrack l'_n, u'_n \rbrack$. For each $v_j$ we have $\tass{r(v_j,l_j,u_j)}$ representing $v_j \in \lbrack l_j, u_j \rbrack$. Then, the binary nogoods $\{\tass{r(v_j,l_j,u_j)}, \fass{r(v_j,l_j-1,u_j)}\}$ and $\{\tass{r(v_j,l_j,u_j)}, \fass{r(v_j,l_j,u_j+1)}\}$ resulting from (6) and (7) are unit, and eventually we get $\tass{r(v_j,l'_j,u'_j)}$ for $\lbrack l_j, u_j \rbrack \subseteq \lbrack l'_j, u'_j \rbrack$. But then the nogood $\{ \tass{r(v_1, l'_1,u'_1)}, \dots, \tass{r(v_n, l'_n, u'_n)} \}$ encoding the conflict region is unit and forces $\fass{r(v_i,l'_i,u'_i)}$ representing $v_i \not\in \lbrack l'_i, u'_i \rbrack$. By nogoods resulting from (6) and (7) we get $d_i$ is not in the domain of $v_i$, and the domains are bound consistent as required.
Since at least one value must be in each domain, encoded in (5), we have a set of non-empty domains which are range consistent.
\end{proof}
A propagator for the global \emph{all-different} constraint that enforces range consistency pruning Hall intervals has been proposed in \citeN{le96a} and encoded to SAT in \citeN{bekanaquwa09a}.
An interval~$\lbrack l, u \rbrack$ is a \emph{Hall interval} iff $|\{ v \mid dom(v) \subseteq \lbrack l, u \rbrack \}| = u - l + 1$. In other words, a Hall interval of size~$k$ completely contains the domains of $k$~variables. Observe that in any bound support, the variables whose domains are contained in the Hall interval consume all values within the Hall interval, whilst any other variable must find their support outside the Hall interval.
\begin{example}
Consider the global \emph{all-different} constraint over the variables~$\{v_1, v_2, v_3, v_4\}$ with $\domain{v_1} = \{2,3\}$, $\domain{v_2} = \{1,2,4\}$, $\domain{v_3} = \{2,3\}$, $\domain{v_4} = \{1,2,3,4\}$. $\lbrack 2,3\rbrack$ is a Hall interval of size 2 as the domain of 2 variables, $v_1$ and $v_3$, is completely contained in it. Therefore we can remove $\lbrack 2,3\rbrack$ from the domains of all the other variables. This leaves $v_2$ and $v_4$ with a domain containing values 1 and 4.
\end{example}
The following decomposition of the global \emph{all-different} constraint will permit us to achieve range consistency via unit propagation. It ensures that no interval $\lbrack l, u\rbrack$ can contain more variables than its size. 
\[
\begin{array}{lr@{\ \leftarrow\ }l}
(8) & & u-l+2\ \{ r(v_1, l, u), \dots, r(v_n, l, u) \}
\end{array}
\]
This simple decomposition can simulate a complex propagation algorithm like \citeANP{le96a}'s with a similar overall complexity of reasoning.
\begin{corollary} \label{col:rng}
Unit-propagation on (5--8) enforces range consistency on the global \emph{all-different} constraint in $\mathcal{O}(nd^3)$ down any branch of the search tree.
\end{corollary}
\begin{proof}
Clearly, the cardinality constraints (8) reflect all conflict regions such that no Hall interval $\lbrack l, u \rbrack$ can contain $u-l+2$ variables, that are more variables than its size. Hence, (8) is a range encoding of the global \emph{all-different} constraint. By Theorem~\ref{thm:rng}, unit-propagation on this encoding enforces range consistency on the global \emph{all-different} constraint.

There are $\mathcal{O}(nd^2)$ nogoods resulting from (5--7) that can be woken $\mathcal{O}(d)$ times down any branch of the search tree. Each propagation requires $\mathcal{O}(1)$ time. Rules (5--7) therefore take $\mathcal{O}(nd^3)$ down any branch of the search to propagate. There are $\mathcal{O}(nd^2)$ nogoods resulting from (8) that each take $\mathcal{O}(1)$ time to propagate down any branch of the search tree. The total running time is given by $\mathcal{O}(nd^3) + \mathcal{O}(nd^2) = \mathcal{O}(nd^3)$.
\end{proof}

\subsection{Bound Encoding}

A last encoding is called the \emph{bound encoding} \cite{crba94a}. In the bound encoding, a propositional variable~$b(v, i)$ is introduced for each value $i$ to represent that the value of~$v$ is bounded by~$i$. That is, $v \leq i$ if $\tass{b(v,i)}$, and $v > i$ if $\fass{b(v,i)}$. Similar to the direct encoding, for each $v$, the truth-assignments of atoms~$b(v, i)$ are encoded by a choice rule (9). In order to ensure that assignments represent a consistent set of bounds, the condition $v \leq i \Rightarrow v \leq i+1$ is posted as integrity constraints (10). Another integrity constraint (11) encodes $v \leq d$, that at least one value must be assigned to $v$:
\[
\begin{array}{cr@{\ \leftarrow\ }l}
(9) & \{ b(v, 1), \dots, b(v, d) \} & \\
(10) & &  b(v, i), \dneg b(v, i+1)\qquad \forall i \in \lbrack 1, d-1 \rbrack \\
(11)& & \dneg b(v, d)
\end{array}
\]
Constraints are encoded into integrity constraints representing conflict regions similar to the range encoding. When all combinations in the region
\[
l_1 < v_1 \leq u_1, \dots, l_n < v_n \leq u_n
\]
violate a constraint, the following rule is added:
\[
\leftarrow b(v_1, u_1), \dots, b(v_n, u_n), \dneg b(v_1, l_1), \dots, \dneg b(v_n, l_n)
\]
\begin{theorem} \label{thm:bnd}
Unit-propagation on the bound encoding enforces bound consistency on the original constraint.
\end{theorem}
\begin{proof}
Suppose we have a set of bounds on the domains of the constraint variables in which no unit-propagation is possible and no domain is empty.
Consider any constraint variable~$v_i$ such that if $v_i$ is assigned its minimum domain value~$l_i+1$ or its maximum domain value~$u_i$ there are no compatible values of the other constraint variables $v_1,\dots, v_{i-1}, v_{i+1}, \dots, v_n$ between their minimum $l_1+1, \dots, l_{i-1}+1, l_{i+1}+1, \dots, l_n+1$ and their maximum domain values $u_1, \dots, u_{i-1}, u_{i+1}, \dots, u_n$, respectively. First, we analyse the case $v_i = u_i$, that is, all instantiations such that $v_i = u_i$ are in a conflict region $l'_1 \leq l_1 < v_1 \leq u_1 \leq u'_1, \dots, l'_n \leq l_n < v_n \leq u_n \leq u'_n$. For each $v_j$ we have $\fass{b(v_j,l_j)}$ and $\tass{b(v_j,u_j)}$, representing $v_j > l_j$ and $v_j \leq u_j$.
Then, the binary nogoods $\{\tass{b(v_j,l_j-1)}, \fass{b(v_j,l_j)}\}$ and $\{\tass{b(v_j,u_j)}, \fass{b(v_j,u_j+1)}\}$ resulting from~(10) are unit, and eventually we get $\fass{b(v_j,l'_j)}$ for $l'_j \leq l_j$ as well as $\tass{b(v_j,u'_j)}$ for $u'_j \geq u_j$. But then the nogood $\{ \fass{b(v_1, l'_1)}, \tass{b(v_1, u'_1)}, \dots, \fass{b(v_n, l'_n)}, \tass{b(v_n, u'_n)} \}$ encoding the conflict region is unit and forces $\tass{b(v_i,l'_i)}$ representing $v_i \leq l'_i$.
Since $l'_i < u_i$ and by the nogoods resulting from~(10) we get $u_i$ is not in the domain of $v_i$.

The second case, where $v_i$ is assigned its minimum domain value $l_i+1$, is symmetric, and we conclude that the domains are bound consistent as required.
Since at least one value must be in each domain, resulting from (11), we have a set of non-empty domains which are bound consistent.
\end{proof}
In order to get a representation of the global \emph{all-different} constraint that can only prune bounds, the bound encoding for variables is linked to (8) as follows:
\[
\begin{array}{cr@{\ \leftarrow\ }l}
(12) & r(v, l, u) & \dneg b(v, l-1), b(v, u) \\
(13) & & r(v, l, u), b(v, l-1) \\
(14) & & r(v, l, u), \dneg b(v, u)
\end{array}
\]
\begin{corollary}
Unit-propagation on (8--14) enforces bound consistency on the global \emph{all-different} constraint in $\mathcal{O}(nd^2)$ down any branch of the search tree.
\end{corollary}
\begin{proof} This result follows from Corollary \ref{col:rng} and Theorem \ref{thm:bnd}. Observe that the decompositions for range and bound consistency both encode the same conflict regions.

For each of the $n$ variables, there are $\mathcal{O}(d)$ nogoods resulting from (9--11) that can be woken $\mathcal{O}(d)$ times down any branch of the search tree. Each propagation requires $\mathcal{O}(1)$ time. Rules (9--11) therefore take $\mathcal{O}(nd^2)$ down any branch of the search to propagate. Furthermore, there are $\mathcal{O}(nd^2)$ nogoods resulting from (8) and (12--14) that each take $\mathcal{O}(1)$ time to propagate down any branch of the search tree. The total running time is given by $\mathcal{O}(nd^2)$.
\end{proof}
Note that an upper bound $h$ can be posted on the size of Hall intervals. The resulting encoding with only those cardinality constraints (5) for which $u - l + 1 \leq h$ detects Hall intervals of size at most $h$, and therefore enforces a weaker level of consistency.

\section{Experiments} \label{sec:exp}

To evaluate our decompositions, we conducted experiments on encodings\footnote[1]{\texttt{http://potassco.sourceforge.net/}} of CSP containing \emph{all-different} and \emph{permutation} constraints.
The global \emph{permutation} constraint is a special case of \emph{all-different} when the number of variables is equal to the number of all their possible values. A decomposition of \emph{permutation} extends (4) by
\[
\leftarrow \dneg e(v_1, i), \dots, \dneg e(v_n, i)
\]
or (8) by the following rule where $1 \leq l \leq u \leq k$:
\[
\leftarrow d-u+l\ \{ \dneg r(v_1, l, u), \dots, \dneg r(v_n, l, u) \}
\]
This can increase propagation.
Our translational approach to Constraint Answer Set Solving has been implemented within the prototypical preprocessor \systemname{inca}.
Although our semantics is propositional, \systemname{inca} compiles constraint logic programs with first-order variables, function symbols, and aggregates, etc. in linear time and space, such that the logic program can be obtained by a \emph{grounding} process.
Experiments consider \systemname{inca}\footnotemark[1] in different settings using different decompositions. We denote the support encoding of the global constraints by \encsup, the bound encoding of the global constraints by \encbou, and the range encoding of the global constraints by \encran. To explore the impact of small Hall intervals, we also tried \encbouh{k} and \encranh{k}, an encoding of the global constraints with only those cardinality constraints (8) for which $u-l+1 \leq k$. The consistency achieved by \encbouh{k} and \encranh{k} is therefore weaker than full bound and range consistency, respectively.

We also include the integrated systems \systemname{clingcon}\footnotemark[1] (0.1.2), and \systemname{ezcsp}\footnote[2]{\texttt{http://krlab.cs.ttu.edu/\~{}marcy/ezcsp/}} (1.6.9) in our empirical analysis.
While \systemname{clingcon} extends the ASP system \systemname{clingo}\footnotemark[1] (2.0.2) with the generic constraint solver \systemname{gecode}\footnote[3]{\texttt{http://www.gecode.org/}} (2.2.0), \systemname{ezcsp} combines the grounder \systemname{gringo}\footnotemark[1] (2.0.3) and ASP solver \systemname{clasp}\footnotemark[1] (1.3.0) with \systemname{sicstus}\footnote[4]{\texttt{http://www.sics.se/sicstus/}} (4.0.8) as a constraint solver.
Note that \systemname{clingo} stands for \systemname{clasp} on \systemname{gringo} and combines both systems in a monolithic way. Since \systemname{inca} is a pure preprocessor, we select the ASP system \systemname{clingo} (2.0.3) as its backend to provide a representative comparison with \systemname{clingcon} and \systemname{ezcsp}.
The behaviour of \systemname{xpanda} is simulated by setting \encsup{} and, therefore, is not considered in our study. We also do not separate time spend on grounding and solving the problem, since the grounder's share of the overall runtime is generally insignificant on our benchmarks.

To compare the performance of Constraint Answer Set solvers against pure CP systems, we also report results of \systemname{gecode} (3.2.0). Its heuristic for variable selection was set to a smallest domain as in \systemname{clingcon}.
All experiments were run on a 2.00~GHz PC under Linux. We report results in seconds, where each run was limited to 600 s time and 1 GB RAM.

\paragraph{Space Complexity.} Data on the size of selected translations shown in Table \ref{tab:timespace} confirms our theoretical results. For $PHP$ and $DW_n$ instances (description follows), the number of atoms in the support (bound/range) encoding is bounded by $\mathcal{O}(n^2)$ ($\mathcal{O}(n^3)$). The number of rules is $\mathcal{O}(n)$ ($\mathcal{O}(n^3)$) for $PHP$, $\mathcal{O}(n^3)$ ($\mathcal{O}(n^3)$) for $DW_n$ due to constraints represented in the direct encoding. An $n \times n$ table is modelled in $QG5$, raising the number of atoms to $\mathcal{O}(n^3)$ ($\mathcal{O}(n^4)$), and rules to $\mathcal{O}(n^4)$ ($\mathcal{O}(n^4)$).
\begin{table}
\caption{Data on time and space for selected translations.}
\label{tab:timespace}
\begin{minipage}{\textwidth}
\begin{tabular}{ccccccccc} \hline\hline
&     & \multicolumn{3}{c}{\encsup} & \multicolumn{2}{c}{\encbou} & \multicolumn{2}{c}{\encran} \\
& $n$ & time & atoms & rules & atoms & rules & atoms & rules \\ \hline
      & 12 & 0.1 & 236 & 151 & 1,829 & 1,907 & 1,061 & 2,423 \\
$PHP$ & 14 & 0.1 & 304 & 177 & 2,876 & 2,981 & 1,630 & 3,877 \\
      & 16 & 0.1 & 380 & 203 & 4,263 & 4,399 & 2,375 & 5,823 \\ \noalign{\vspace {.2cm}}

      &  8 & 0.3 &   655 &  36,525 & 1,297 &  41,565 &  3,217 &  38,237 \\
$QG5$ & 10 & 0.9 & 1,219 & 109,066 & 2,421 & 118,946 &  7,121 & 111,446 \\
      & 12 & 2.3 & 2,039 & 267,643 & 4,057 & 284,755 & 13,849 & 270,067 \\ \noalign{\vspace {.2cm}}

      &  4 & 0.6 &   564 &    74,945 &    941 &    78,848 &  4,135 &    81,174 \\
$DW_n$&  6 & 3.0 & 1,130 &   363,115 &  1,983 &   373,002 & 12,581 &   381,724 \\
      &  8 & 9.2 & 1,888 & 1,122,549 &  3,409 & 1,142,132 & 28,355 & 1,163,810 \\ \hline\hline
\end{tabular}
\vspace{-2\baselineskip}
\end{minipage}
\end{table}

\subsection{Pigeon Hole Problems}

The \emph{Pigeon Hole Problem} (PHP) is to show that it is impossible to put $n$ pigeons into $n-1$ holes if each pigeon must be put into a distinct hole.
Clearly, our bound and range decompositions are faster compared to weaker encodings (see Table \ref{tab:php}). On such problems, detecting large Hall intervals is essential.
\begin{table}
\caption{Runtime results in seconds for PHP.}
\label{tab:php}
\begin{minipage}{\textwidth}
\begin{tabular}{ccccccccccc} \hline\hline
$n$ & \encsup & \encbouh{1} & \encbouh{2} & \encbouh{3} & \encbou & \encranh{3} & \encran & \systemname{ezcsp} & \systemname{clingcon} & \systemname{gecode} \\ \hline
10 & 5.4 & 0.7 & 0.5 & 0.1 & \textbf{0.0} & 0.2 & \textbf{0.0} & 1.8 & 1.4 & 0.9 \\
11 & 46.5 & 3.5 & 1.5 & 1.0 & \textbf{0.0} & 1.9 & \textbf{0.0} & 16.7 & 15.2 & 9.0 \\
12 & 105.0 & 14.8 & 7.1 & 3.9 & \textbf{0.0} & 2.6 & 0.1 & 183.9 & 172.5 & 104.1 \\
13 & --- & 91.4 & 68.6 & 25.4 & 0.1 & 30.4 & \textbf{0.0} & --- & --- & --- \\
14 & --- & --- & 350.1 & 125.0 & \textbf{0.0} & 196.9 & 0.1 & --- & --- & --- \\
15 & --- & --- & --- & --- & \textbf{0.1} & --- & \textbf{0.1} & --- & --- & --- \\
16 & --- & --- & --- & --- & \textbf{0.1} & --- & \textbf{0.1} & --- & --- & --- \\ \hline\hline
\end{tabular}
\vspace{-2\baselineskip}
\end{minipage}
\end{table}

\subsection{Quasigroup Completion}

A \emph{quasigroup} is an algebraic structure~$(Q,\cdot)$, where $Q$ is a set and~$\cdot$ is a binary operation on~$Q$ such that for every pair of elements~$a,b \in Q$ there exist unique elements $x,y \in Q$ which solve the equations $a \cdot x = b$ and $y \cdot a = b$. The \emph{order}~$n$ of a quasigroup is defined by the number of elements in~$Q$. A quasigroup can be represented by an $n \times n$-multiplication table, where for each pair $a,b$ the table gives the result of $a \cdot b$, and it defines a \emph{Latin square}. This means that each element of $Q$ occurs exactly once in each row and each column of the table. 
The \emph{Quasigroup Completion Problem}~(QCP) is to determine whether a partially filled table can be completed in such a way that a multiplication table of a quasigroup is obtained.
Randomly generated QCP has been proposed as a benchmark domain for CP systems by \citeN{gose97a} since it combines the features of purely random problems and highly structured problems.
\begin{table}
\caption{Average times over 100 runs on QCP. Timeouts are given in parenthesis.}
\label{tab:qcp}
\begin{minipage}{\textwidth}
\begin{tabular}{cccccccccc} \hline\hline
\% & \encsup & \encbouh{3} & \encbou & \encranh{3} & \encran & \systemname{ezcsp} & \systemname{clingcon} & \systemname{gecode} & \systemname{gecode}$_{BC}$ \\ \hline
10 & 2.6 & 5.0 & 8.2 & 6.0 & 7.3 &29.6 (7) & 9.7 (4) & 2.2 (4) & 0.5 (1) \\
20 & 2.4 & 5.0 & 8.0 & 6.2 & 7.2 &21.3 (20) & 6.2 (5) & 5.0 (4) & 0.9 (3) \\
30 & 2.3 & 4.8 & 7.9 & 6.1 & 7.1 &10.3 (30) & 12.9 (13) & 2.9 (13) & 1.1 (5) \\
35 & 2.3 & 4.8 & 7.9 & 6.1 & 7.0 &21.6 (24) & 11.2 (17) &14.1 (13) & 6.2 (7) \\
40 & 2.3 & 4.7 & 7.8 & 6.0 & 6.9 &51.6 (29) & 23.1 (22) &11.7 (20) & 5.7 (9) \\
45 & 2.3 & 4.7 & 7.8 & 5.9 & 6.8 &36.3 (35) & 14.7 (28) &17.7 (25) & 6.3 (13) \\
50 & 2.3 & 4.6 & 7.7 & 5.9 & 6.8 &36.1 (50) & 21.2 (37) &25.1 (32) & 6.3 (18) \\
55 & 2.3 & 4.5 & 7.6 & 5.8 & 6.7 &61.4 (51) & 24.4 (44) &19.6 (41) &30.9 (29) \\
60 & 2.2 & 4.4 & 7.5 & 5.6 & 6.6 &60.2 (63) & 31.4 (56) &36.0 (51) &27.2 (35) \\
70 & 2.2 & 4.2 & 7.1 & 5.1 & 6.0 &70.0 (66) & 30.2 (50) &28.0 (45) &17.0 (27) \\
80 & 2.1 & 4.0 & 6.7 & 4.7 & 5.5 &16.2 (18) & 4.2 (18) &17.2 (13) & 7.0 (7) \\
90 & 2.1 & 4.0 & 6.7 & 4.7 & 5.5 & 1.4 & 2.6 (1) & 0.4 (1) & 3.2 \\ \hline\hline
\end{tabular}
\vspace{-2\baselineskip}
\end{minipage}
\end{table}
Table \ref{tab:qcp} compares the runtime for solving QCP problems of size $20 \times 20$ where the first column gives the percentage of preassigned values. We included \systemname{gecode} with algorithms that enforce bound and domain consistency, denoted as \systemname{gecode}$_{BC}$ and \systemname{gecode}$_{DC}$ (not shown), in the experiments. Our analysis exhibits phase transition behaviour of the systems \systemname{ezcsp}, \systemname{clingcon}, \systemname{gecode}, and \systemname{gecode}$_{BC}$, while our Boolean encodings and \systemname{gecode}$_{DC}$ solve all problems within seconds. Interestingly, learning constraint interdependencies as in our approach is sufficient to tackle QCP. In fact, most of the time for $S$, $B_k$, $R_k$ is spent on grounding, but not for solving the actual problem.

\subsection{Quasigroup Existence}

The \emph{Quasigroup Existence Problem} (QEP) is to determine the existence of certain interesting classes of quasigroups. We follow \citeN{fuslbe93a} and look at problems QG1 to QG7 that were target to open questions in finite mathematics.
We represent them in the direct encoding which weakens the overall consistency. Furthermore,
we add the axiom $a \cdot n \geq a - 1$ where $n$ is the order of the desired quasigroup, to avoid some symmetries in search space. We also assume quasigroups to be \emph{idempotent}, that means $a \cdot a = a$. QEP has been proposed as a benchmark domain for CP systems in \citeN{gewa99a}.
\begin{table}
\caption{Runtime results in seconds for QEP.}
\label{tab:qep}
\begin{minipage}{\textwidth}
\begin{tabular}{ccccccccccc} \hline
 & $n$ & \encsup & \encbouh{1} &
 \encbouh{3} & \encbou & \encranh{3} & \encran & \systemname{ezcsp} & \systemname{clingcon} & \systemname{gecode} \\ \hline\hline
& $7$ & 1.7 & 1.7 & 
 1.7 & 1.7 & 1.7 & 1.6 & 65.0 & 189.8 & \textbf{0.6} \\
QG1 & $8$ & 19.0& 5.9 &
 \textbf{4.7} &19.8 & 6.4 &\textbf{4.7} & --- & --- & --- \\
& $9$ & --- & 139.4 &
 152.0& 234.6 & \textbf{27.6} &466.9 & --- & --- & --- \\ \noalign{\vspace {.2cm}}
& $7$ & 1.7 & 1.7 &
 1.7 & 1.8 & 1.7 & 1.8 & 46.1 & 1.5 & \textbf{1.2} \\
QG2 & $8$ & 46.6& \textbf{9.6} &
10.6 &37.7 &11.7 &14.8 & --- & --- & --- \\
& $9$ & --- & 246.0 &
 \textbf{55.7}& 88.3 &119.7 &213.4 & --- & --- & --- \\ \noalign{\vspace {.2cm}}
& $7$ & 0.2 & 0.2 &
 0.2 & 0.3 & 0.2 & 0.3 & 3.2 & 1.0 & \textbf{0.0} \\
QG3 & $8$ & 0.4 & 0.4 &
 0.5 & 0.5 & 0.5 & 0.5 & 4.3 & 9.0 & \textbf{0.2} \\
& $9$ &10.2 &\textbf{7.4} &
 9.5 &16.5 &11.0 &12.8 & --- & --- & 18.2 \\ \noalign{\vspace {.2cm}}
& $7$ & 0.2 & 0.2 &
 0.2 & 0.3 & 0.3 & 0.3 & 2.8 & 0.7 & \textbf{0.1} \\
QG4 & $8$ & 0.5 & 0.6 &
 0.7 & 0.9 & 0.8 & 0.7 &27.9 &36.8 & \textbf{0.3} \\
& $9$ & 1.3 & \textbf{1.0} &
 2.1 &3.0 &1.1 &0.9 &442.1&288.8& 3.7 \\ \noalign{\vspace {.2cm}}
& $8$ & 0.4 & 0.4 &
 0.4 & 0.5 & 0.4 & 0.4 & 6.9 & 5.3 & \textbf{0.0} \\
& $9$ & 0.7 & 0.8 &
 0.8 & 0.9 & 0.8 & 0.8 &249.2& --- & \textbf{0.0} \\
QG5 & $10$& 1.6 & 1.5 &
 1.6 & 1.9 & 1.6 & 1.6 & --- & --- & \textbf{0.2} \\
& $11$& 2.1 & 2.2 &
 2.4 & 3.4 & 2.8 & 2.4 & --- & --- & \textbf{0.8} \\
& $12$&27.0 &\textbf{6.2}&
 9.1 &12.4 & 8.4 &10.4 & --- & --- & 16.4 \\ \noalign{\vspace {.2cm}}
& $8$ & 0.4 & 0.4 &
 0.5 & 0.5 & 0.5 & 0.4 & 0.8 & --- & \textbf{0.0} \\
& $9$ & 0.7 & 0.7 &
 0.8 & 0.9 & 0.8 & 0.8 & 1.2 & --- & \textbf{0.0} \\
QG6 & $10$& 1.2 & 1.4 &
 1.5 & 1.8 & 1.6 & 1.5 &10.5 & --- & \textbf{0.1} \\
& $11$& 2.7 & 2.8 &
 4.0 & 4.2 & 3.9 & 4.8 &125.5& --- & \textbf{1.2} \\
& $12$&32.0 &\textbf{12.9}&
25.6 &36.4 &25.7 &50.6 & --- & --- & 24.6 \\ \noalign{\vspace {.2cm}}
& $8$ & 0.4 & 0.4 &
 0.4 & 0.6 & 0.5 & 0.5 & 1.1 & --- & \textbf{0.1} \\
QG7 & $9$ & \textbf{0.7} & 1.0 &
 1.2 & 1.7 & 1.2 & 1.4 & 9.1 & --- & 0.9 \\
& $10$ & 6.7 & \textbf{3.2} &
 5.2 & 8.0 & 4.7 & 4.6 & --- & --- & 22.0 \\ \hline\hline
\end{tabular}
\vspace{-2\baselineskip}
\end{minipage}
\end{table}
All axioms have been modelled in \systemname{ezcsp} and \systemname{gecode} using constructive disjunction, and in $S$, $B_k$, $R_k$ and \systemname{clingcon} using integrity constraints.
Table \ref{tab:qep} demonstrates that both constructive disjunction and integrity constraints have a similar behaviour, as for \systemname{ezcsp} and \systemname{clingcon} on benchmark classes QG1 to QG4. On harder instances, 
conflict-driven learning appears to be too costly for \systemname{clingcon}. Additional experiments revealed that \systemname{clingcon} without learning performs like \systemname{ezcsp}. 
On the other hand, our decompositions benefit from learning constraint interdependencies, resulting in runtimes that outperform all other systems including \systemname{gecode} on the hardest problems.

\subsection{Graceful Graphs}

A labelling $f$ of the nodes of a graph $(V,E)$ is \emph{graceful} if $f$ assigns a unique label~$f(v)$ from $\{0,1,\dots,|E|\}$ to each node $v \in V$ such that, when each edge $(v,w) \in E$ is assigned the label $|f(v)-f(w)|$, the resulting edge labels are distinct. The problem of determining the existence of a graceful labelling of a graph (GGP) has been modelled as a CSP in \citeN{pesm03a}. Our experiments consider double-wheel graphs $DW_n$ composed by two copies of a cycle with $n$ vertices, each connected to a central hub.
\begin{table}
\caption{Runtime results in seconds for GGP.}
\label{tab:ggp}
\begin{minipage}{\textwidth}
\begin{tabular}{cccccccccc} \hline\hline
$DW_n$ & \encsup & \encbouh{1} & \encbouh{3} & \encbou & \encranh{3} & \encran & \systemname{ezcsp} & \systemname{clingcon} & \systemname{gecode} \\ \hline
$3$ & 11.4 &  3.8 &  5.7 &  8.7 &  6.0 & 10.4 & 6.5 & 66.9 & \textbf{1.8} \\
$4$ &  1.3 &  2.0 &  1.5 &  3.2 &  3.0 &  2.5 & 0.6 & \textbf{0.1} & \textbf{0.1} \\
$5$ &  4.5 &  5.0 &  4.5 & 13.5 & 12.5 & 31.4 & 1.0 & 2.0 & \textbf{0.1} \\
$6$ &  7.2 & 11.0 & 17.6 & 47.7 & 21.3 &110.2 & \textbf{1.2} & --- & 7.2 \\
$7$ & 23.8 & 28.3 & 67.9 &227.9 & 60.0 &432.9 & \textbf{18.0} & --- & --- \\
$8$ & 48.4 & 68.4 & ---  &207.8 & 58.4 &356.8 & \textbf{4.3} & --- & --- \\
$9$ & \textbf{82.8} &106.5 &200.4 &486.6 &227.4 & ---  & 390.5 & --- & --- \\ \hline\hline
\end{tabular}
\vspace{-2\baselineskip}
\end{minipage}
\end{table}
Table \ref{tab:ggp} shows that our encodings compete with \systemname{ezcsp} and outperform the other comparable systems, where the support encoding performs better than bound and range encodings. In most cases, the branching heuristic used in our approach appears to be misled by the extra variables introduced in $B_k$ and $R_k$. That explains some of the variability in the runtimes.

\section{Conclusions} \label{sec:con}

We have provided a new translation-based approach to incorporating Constraint Processing into Answer Set Programming. In particular, we investigated various generic ASP decompositions for constraints on finite domains and proved which level of consistency unit-pro\-pa\-ga\-tion achieves on them.
Our techniques were formulated as preprocessing and can be applied to any ASP system without changing its source code, which allows for programmers to select the solvers that best fit their needs. We have empirically evaluated their performance on benchmarks from CP and found them outperforming integrated Constraint Answer Set Programming systems as well as pure CP solvers.
As a key advantage of our novel approach we identified CDNL, exploiting constraint interdependencies which can improve propagation between constraints.
Future work concerns a comparison to \citeANP{niemela99a}'s encoding \cite{niemela99a,yoho04a}, and encodings of further global constraints useful in Constraint Answer Set Programming.
\paragraph{Acknowledgements.} We are grateful to Martin Gebser and Torsten Schaub for useful discussions on the subject of this paper.

\end{document}